\documentclass[reqno, eucal]{amsart}

\usepackage[mathscr]{eucal} 
\usepackage[dvips]{graphicx}
\usepackage[dvips]{color}
\usepackage{amsmath,amsfonts,amssymb,amsthm,amscd}
\usepackage{epic}
\usepackage{eepic}
\usepackage{longtable}
\usepackage{array}

\numberwithin{equation}{section}

\newtheorem{theorem}{Theorem}[section]

\newtheorem{proposition}[theorem]{Proposition}

\theoremstyle{definition}

\newtheorem{example}[theorem]{Example}

\newcommand{\Z}{{\mathbb Z}}
\newcommand{\R}{{\mathbb R}}
\newcommand{\C}{{\mathbb C}}

\newcommand{\ok}{{\rm{\bf k}}}
\newcommand{\OK}{{\rm{\bf K}}}
\newcommand{\am}{{\rm{\bf a}}^{\!-} }
\newcommand{\ap}{{\rm{\bf a}}^{\!+} }
\newcommand{\apm}{{\rm{\bf a}}^{\!\pm} }
\newcommand{\amp}{{\rm{\bf a}}^{\!\mp} }
\newcommand{\Am}{{\rm{\bf A}}^{\!-} }
\newcommand{\Ap}{{\rm{\bf A}}^{\!+} }
\newcommand{\ichi}{{\bf 1} }

\begin{document}

\title[3D reflection equation]{
A polynomial formula for the solution of \\
3D reflection equation}

\author{Atsuo Kuniba}
\email{atsuo@gokutan.c.u-tokyo.ac.jp}
\address{Institute of Physics, University of Tokyo, Komaba, Tokyo 153-8902, Japan}

\author{Shouya Maruyama}
\email{maruyama@gokutan.c.u-tokyo.ac.jp}
\address{Institute of Physics, University of Tokyo, Komaba, Tokyo 153-8902, Japan}

\maketitle

\begin{center}{\bf Abstract}
\end{center}

We introduce a family of polynomials in $q^2$ and four variables associated with 
the quantized algebra of functions $A_q(C_2)$.
A new formula is presented for the recent solution of the 3D reflection equation
in terms of these polynomials specialized to the eigenvalues of 
the $q$-oscillator operators.
 
\vspace{0.4cm}

\section{Introduction}\label{sec:intro}

In quantum integrable systems in 3 dimension (3D),
an important role is played by 
the Zamolodchikov tetrahedron equation \cite{Zam80}
and the Isaev-Kulish 3D reflection equation \cite{IK}:
\begin{align}
{R}_{356}{R}_{246}{R}_{145}{R}_{123}
&={R}_{123}{R}_{145}{R}_{246}{R}_{356}, \label{a:teq}\\
{R}_{456}{R}_{489}
{K}_{3579}{R}_{269}{R}_{258}
{K}_{1678}{K}_{1234}
&={K}_{1234}
{K}_{1678}
{R}_{258}{R}_{269}
{K}_{3579}{R}_{489}{R}_{456}.\label{a:3dref}
\end{align}
They are equalities among linear operators acting on the 
tensor product of 6 and 9 vector spaces, respectively.
The indices specify the components in the tensor product 
on which the operators $R$ and $K$ act nontrivially.
They serve as a 3D analogue of the Yang-Baxter \cite{Bax} 
and the reflection equations 
postulating certain factorization conditions on straight strings
which undergo the scattering $R$ and the reflection $K$ by a boundary plane.
We call the scattering and the reflection operators 
as 3D $R$ and 3D $K$ for short.

The first nontrivial solution $K$ to the 3D reflection equation (\ref{a:3dref})
was constructed in \cite{KO1} based on the representation theory \cite{So2}
of the quantized algebra of functions $A_q(C_n)$ \cite{RTF}.
It is essentially obtained as the intertwiner 
of the two equivalent irreducible $A_q(C_2)$ modules
${\mathcal F}_{q^2}\otimes{\mathcal F}_{q}\otimes
{\mathcal F}_{q^2}\otimes{\mathcal F}_{q}
\simeq {\mathcal F}_{q}\otimes
{\mathcal F}_{q^2}\otimes{\mathcal F}_{q}\otimes
{\mathcal F}_{q^2}$,
where ${\mathcal F}_{q}, {\mathcal F}_{q^2}$ are the Fock spaces of the 
$q$-oscillators 
$\langle {{\rm{\bf a}}^{\!\pm} }, \ok\rangle,
\langle {{\rm{\bf A}}^{\!\pm} }, \OK\rangle$.
See (\ref{akf}). 
The $K$ is also characterized as the transition matrix of the 
Poincar\'e-Birkhoff-Witt (PBW) bases of the positive part of 
$U_q(C_2)$ \cite{KOY}. See (\ref{eke}).
Matrix elements of the $K$ are polynomials in $q$
whose $q\rightarrow 0$ limits are still known to yield a decent 
set-theoretical solution to the 3D reflection equation \cite{KO1}.
So far a general formula for the $K$ is only available  
in \cite[Th.3.4]{KO1},
which consists of sums of ratios of many $q$-factorials.
See (\ref{mess}). 

The status contrasts with the relevant 3D $R$, 
the companion object in (\ref{a:3dref}), 
for which a number of results have been established.
It was originally obtained as the intertwiner of the 
quantized algebra of functions $A_q(A_2)$ \cite{KV} and found 
later also in a quantum geometry consideration \cite{BS}.
They were shown to be the same object 
in \cite{KO1}.
Several formulas are available including 
\cite[eq.(30)]{BS}, \cite[eq.(58)]{BMS}, 
\cite[eq.(2.20)]{KO1} (a correction of the misprint in \cite[p194]{KV}),
\cite[eq. (104)]{S}  
and  \cite[eq.(9)]{KO3}.

The purpose of this paper is to provide the relatively intact 
3D $K$ with the new explicit formulas 
(\ref{Kq}) and (\ref{Kqq}),  
which are more structural than the previous one in \cite{KO1}.
We introduce a family of polynomials 
$\{Q_{b,c}(x,y,z,w)\mid (b,c) \in \Z_{\ge 0}^2\}$ with 
coefficients in $\Z[q^2]$ that are
characterized by a system of $q$-difference equations. 
The equations are polynomial forms of the 
intertwining relations for the 3D $K$ where the 
four variables $x,y,z,w$ correspond to the four positive roots of $C_2$.
See (\ref{proot}).
The elements of the $K$ are expressed as a specialization of the 
polynomials to the eigenvalues of the $q$-oscillator operators 
$\ok$ and $\OK$.
See Example \ref{ex:k}.

Our new formula may be viewed as a generalization of the 
analogous result on the 3D $R$ in \cite{BS,BMS}.
It is still a cumbersome expression reflecting a significantly more involved 
nature of the $K$ compared with the $R$.
However extracting the polynomial structure implies an 
``analytic continuation" of the eigenvalues of $\ok, \OK$ to generic variables, 
which is an important step toward 
a possible extension to the modular double setting. 
A more account on this will be given in Section \ref{sec:sum}.
See also \cite{BMS,KOS} and reference therein for the generalization and application
of the 3D $R$ associated with the modular double. 
We hope to report on this issue elsewhere.

In Section \ref{sec:ori} 
the origin of the 3D $K$ is recalled based on \cite{KO1} and \cite{KOY}.
In Section \ref{sec:pf}  
the polynomials $Q_{b,c}(x,y,z,w)$ are introduced 
and some basic properties are established.
The new formulas for the $K$ are presented with a proof.
In Section \ref{sec:R} a review of the 
analogous result on the 3D $R$ \cite{BS,BMS} is given for comparison.
It also includes supplementary $q$-difference equations, 
an example of the proof of the integral formula 
\cite[eq. (104)]{S} and a derivation of the new formula 
for the matrix elements of the 3D $R$ \cite[eq. (9)]{KO3}.
Section \ref{sec:sum} contains a summary and an outlook.
Appendix \ref{app:de} lists the intertwining relations of 
the $K$ and the corresponding 
$q$-difference equations of $Q_{b,c}(x,y,z,w)$.

We assume that $q$ is generic and use the notation
\begin{align}
&(z;q)_n = \frac{(z;q)_\infty}{(zq^{n};q)_\infty},\quad
(q)_n=(q;q)_n,\quad 
\binom{r_1,\ldots, r_m}{s_1,\ldots, s_n}_{\!\!q} = 
\frac{\prod_{i=1}^m(q)_{r_i}}{\prod_{i=1}^n(q)_{s_i}}\label{cbdef}
\end{align}
in terms of $(z;q)_\infty = \prod_{j\ge0}(1-zq^j)$.
They will only be used for integer indices,
therefore the assumption $|q| <1$ is not necessary, despite the 
formal appearance of the infinite product.
The last symbol will be used without assuming 
$\sum_{i=1}^m r_i = \sum_{i=1}^ns_i$.
An important consequence of these definitions is the support property:
$1/(q)_n = 0$ for $n \in \Z_{<0}$. 
Thus  
$\binom{r_1,\ldots, r_m}{s_1,\ldots, s_n}_{\!\!q}=0$
if $\{r_1,\ldots, r_m\} \cap \Z_{< 0} = \emptyset$ and 
$\{s_1,\ldots, s_n\} \cap \Z_{<0}\neq \emptyset$.

\section{Origin of 3D $K$}\label{sec:ori}

\subsection{$K$ as the intertwiner of $A_q(C_2)$ modules}\label{ss:AK}

We recall the definition of  $A_q(C_n)$ \cite{RTF}, 
where it was denoted by 
$\mathrm{Fun(Sp}_q(n))$.
First introduce the structure constants 
$(R_{ij, kl})_{1\le i,j,k,l \le 2n}$ and 
$C=-C^{-1}=(C_{ij})_{1 \le i,j \le 2n}$ by
\begin{align*}
\begin{split}
&\sum_{i,j,k,l}R_{ij,kl}E_{ik}\otimes E_{jl}=
q\sum_i E_{ii}\otimes E_{ii} + \sum_{i\neq j, j'}E_{ii}\otimes E_{jj}
+q^{-1}\sum_i E_{ii}\otimes E_{i'i'}\\
&\qquad\qquad\qquad+(q-q^{-1})\sum_{i>j}E_{ij}\otimes E_{ji}
-(q-q^{-1})
\sum_{i>j}\epsilon_i\epsilon_jq^{\varrho_i-\varrho_j}E_{ij}\otimes E_{i'j'},
\end{split}\\
&C_{ij}= \delta^i_{j'}\epsilon_i q^{\varrho_j},\quad i' = 2n+1-i,
\quad
\epsilon_i = 1\,(1\le i \le n),\;\;\epsilon_i = -1\,(n<i \le 2n),\\
&(\varrho_1,\ldots, \varrho_{2n}) = (n,n-1,\ldots,1,-1,\ldots, -n+1,-n).
\end{align*}
Here $E_{ij}$ is a matrix unit and the indices are summed over 
$\{1,2,\ldots, 2n\}$ under the specified conditions.
The 
$\sum R_{ij,kl}E_{ik}\otimes E_{jl}$ is a limit of the quantum $R$ matrix for the 
vector representation of $U_q(C_n)$ given in \cite[eq.(3.6)]{J2}.
 
The quantized algebra of functions $A_q(C_n)$
is a Hopf algebra generated by $T = (t_{ij})_{1\le i,j\le 2n}$ 
with the relations symbolically expressed as
$R(T\otimes T) = (T\otimes T)R$ and 
$TCT^tC^{-1} = CT^tC^{-1}T = I$. 
Explicitly they read
\begin{align*}
\sum_{m,p}R_{ij,mp}t_{mk}t_{pl} 
= \sum_{m,p}t_{jp}t_{im}R_{mp,kl},\quad
\sum_{jkl}C_{jk}C_{lm}t_{ij}t_{lk} 
= \sum_{jkl}C_{ij}C_{kl}t_{kj}t_{lm} = -\delta^i_{m},
\end{align*}
where $\delta^i_m=1$ if $i=m$ and $0$ otherwise.
The coproduct is given by 
$\Delta(t_{ij}) = \sum_k t_{ik}\otimes t_{kj}$.
We will use the symbol $\Delta$ to also mean
the multiple coproducts like
$(\Delta\otimes 1) \circ \Delta = (1\otimes \Delta) \circ \Delta$, etc.

To describe the representations of $A_q(C_2)$, we introduce the Fock space
${\mathcal F}_q = \oplus_{m\ge 0}\C(q) |m\rangle$ 
equipped with the $q$-oscillator operators $\ap, \am, \ok$ acting as
\begin{align}\label{akf}
{\rm{\bf k}}|m\rangle = q^m |m\rangle,\;
{\rm{\bf a}}^+|m\rangle = |m+1\rangle,\;
{\rm{\bf a}}^-|m\rangle = (1-q^{2m})|m-1\rangle.
\end{align}
Let ${\mathcal F}_{q^2}, \Ap, \Am$ and  $\OK$ denote 
the corresponding $q$-oscillator operators with $q$ replaced by $q^2$.  
Thus $\OK|m\rangle=q^{2m}|m\rangle,\,
\Ap|m\rangle = |m+1\rangle,\,
\Am|m\rangle = (1-q^{4m})|m-1\rangle$.

Now we consider the $n=2$ case.
Then for $T=(t_{ij})_{1\le i,j \le 4}$ the maps
\begin{align}
&\pi_1(T)= 
\begin{pmatrix}
\mu_1\am & \alpha_1\ok & 0& 0\\
\beta_1\ok& \nu_1\ap  & 0& 0\\
0 & 0 & \nu_1^{-1}\am & q\beta_1^{-1}\ok \\
0 & 0 & q\alpha_1^{-1}\ok & \mu_1^{-1}\ap
\end{pmatrix}\quad
(-q^{-1}\alpha_1\beta_1 = \mu_1\nu_1=\epsilon=\pm 1),\label{pi31}\\
&\pi_2(T)= 
\begin{pmatrix}
\rho{\bf 1}  & 0 & 0 & 0\\
0 & \mu_2\Am & \alpha_2 \OK & 0\\
0 & \beta_2 \OK & \mu_2^{-1}\Ap & 0\\
0 & 0 & 0  & \rho^{-1}{\bf 1}\\
\end{pmatrix}\quad(\rho=\pm 1,\; \alpha_2 \beta_2 = -q^2)\label{pi33}
\end{align}
give the irreducible representations 
$\pi_i: A_q(C_2) \rightarrow \mathrm{End}({\mathcal F}_{q^i})$ \cite{KO1}.
Here ${\bf 1}$ denotes the identity operator on 
${\mathcal F}_{q^2}$.
The parameters $\alpha_i, \beta_i, \mu_i, \nu_i$ 
are to obey the constraints in the parentheses.
According to \cite{So2}, irreducible representations of $A_q(C_2)$ 
are labeled by the elements of the Weyl group $W(C_2)$ 
up to a torus degree of freedom.
The $W(C_2)$ is a Coxeter system generated by the simple reflections
$s_1$ and $s_2$ with the relations
$s_1^2=s_2^2=1,\; s_1s_2s_1s_2 = s_2s_1s_2s_1$.
(We employ the convention such that the indices 1 and 2 correspond to the 
short and the long simple root of $C_2$, respectively.)
The $\pi_i$ is the irreducible $A_q(C_2)$ module corresponding to $s_i$.
The irreducible representation for 
$s_is_js_is_j$  is given by $\pi_i\otimes \pi_j\otimes \pi_i\otimes \pi_j
\,(i\neq j)$.
We write such a tensor product representation as
$\pi_{ijij}$ for short.
The relation $s_1s_2s_1s_2 = s_2s_1s_2s_1$ implies the equivalence 
$\pi_{1212}\simeq \pi_{2121}$.
Therefore there is a unique map
\begin{align*}
\Psi : 
\mathcal{F}_{q} \otimes \mathcal{F}_{q^2}\otimes 
\mathcal{F}_{q}\otimes \mathcal{F}_{q^2} \longrightarrow
\mathcal{F}_{q^2} \otimes \mathcal{F}_{q}
\otimes \mathcal{F}_{q^2}\otimes \mathcal{F}_{q}
\end{align*}
characterized by the intertwining relations and the normalization:
\begin{align}
&\pi_{2121}(\Delta(f))\circ \Psi = \Psi \circ \pi_{1212}(\Delta(f))
\quad (\forall f \in A_q(C_2)),\label{pip}\\
&\Psi (|0\rangle \otimes|0\rangle \otimes|0\rangle \otimes|0\rangle) 
=|0\rangle \otimes|0\rangle \otimes|0\rangle \otimes|0\rangle. \nonumber
\end{align}
We find it convenient to work with $\mathscr{K}$ defined by
\begin{align*}
\mathscr{K} = \Psi  P_{1234} : 
\; \mathcal{F}_{q^2} \otimes \mathcal{F}_{q}\otimes 
\mathcal{F}_{q^2}\otimes \mathcal{F}_{q}
\longrightarrow 
\mathcal{F}_{q^2} \otimes \mathcal{F}_{q}\otimes 
\mathcal{F}_{q^2}\otimes \mathcal{F}_{q},
\end{align*}
where 
$P_{1234} : x_1 \otimes x_2 \otimes x_3 \otimes x_4 \mapsto
x_4 \otimes x_3 \otimes x_2 \otimes x_1$ is the linear operator 
reversing the order of the tensor product.
The intertwining relation (\ref{pip}) is translated into
\begin{align}
&\pi_{2121}(\Delta(f))\circ \mathscr{K} 
= \mathscr{K} \circ \pi_{2121}(\tilde{\Delta}(f))
\quad (\forall f \in A_q(C_2)),\label{reK}
\end{align}
where $\tilde{\Delta}(t_{ij}) = \sum_{l_1,l_2,l_3} 
t_{l_3 j}\otimes t_{l_2 l_3} \otimes t_{l_1l_2} \otimes t_{i l_1}$.

Introduce the matrix elements of $\mathscr{K}$ by 
$\mathscr{K}(|i\rangle \otimes |j\rangle \otimes |k\rangle \otimes |l\rangle) = 
\sum_{a,b,c,d} \mathscr{K}^{a, b, c, d}_{i, j, k, l}
|a\rangle \otimes |b\rangle \otimes |c\rangle \otimes |d\rangle$.
We set 
$K^{a,b,c,d}_{i,j,k,l}=\epsilon^{b+l}\mu_1^{2k-2c}
(\rho \mu_2)^{j-b}\mathscr{K}^{a,b,c,d}_{i,j,k,l}$ 
by using the parameters in (\ref{pi31}) and (\ref{pi33}).
Then it turns out that 
$K^{a,b,c,d}_{i,j,k,l}$ is a polynomial in $q$ free from all the other parameters
\cite{KO1}.
More precisely, 
$K^{a,b,c,d}_{i,j,k,l} \in q^\eta \Z[q^2]$ holds, where 
$\eta=0,1$ is specified by $\eta \equiv bd+jl \mod 2$. 
It satisfies \cite[eq.(3.25)]{KO1} which in the present notation reads
\begin{align}\label{kud}
K^{a,b,c,d}_{i,j,k,l} = \binom{j,l}{b,d}_{\!\!q^2}
\binom{i,k}{a,c}_{\!\!q^4}K^{i,j,k,l}_{a,b,c,d}.
\end{align}
We introduce the (parameter-free) 3D reflection operator 
$K \in \mathrm{End}(\mathcal{F}_{q^2} \otimes \mathcal{F}_{q}\otimes 
\mathcal{F}_{q^2}\otimes \mathcal{F}_{q})$\footnote{
We warn that $K$ and $\mathscr{K}$ here 
are denoted by $\mathscr{K}$ and $K$, respectively in \cite{KO1}.} by
\begin{align}\label{Kact}
K(|i\rangle \otimes |j\rangle \otimes |k\rangle \otimes |l\rangle) = 
\sum_{a,b,c,d} K^{a, b, c, d}_{i, j, k, l}
|a\rangle \otimes |b\rangle \otimes |c\rangle \otimes |d\rangle.
\end{align}

Appendix \ref{app:de} lists the relation (\ref{reK}) for $K$ 
with $f=t_{i-1,j-1}$ as $\langle\text{ij}\rangle$. 
The 3D reflection equation (\ref{a:3dref}) follows from the
equivalence 
$\pi_{123212323} \simeq \pi_{323212321}$ 
of the irreducible representations for $A_q(C_3)$ 
reflecting the two reduced expressions 
$s_1s_2s_3s_2s_1s_2s_3s_2s_3 
= s_3s_2s_3s_2s_1s_2s_3s_2s_1$ 
of the longest element of 
the Weyl group $W(C_3)$.
Further properties of $K$ are available in \cite{KO1}.

\subsection{$K$ as the transition coefficient of PBW bases}\label{ss:pbw}

Let $U^+_q(C_2)$ be the positive part of the 
quantized universal enveloping algebra of $U_q(C_2)$.
It is an associative algebra generated by $e_1$ and $e_2$ 
obeying the $q$-Serre relations:
\begin{align*}
e_1^3e_2-[3]_qe_1^2e_2e_1+[3]_qe_1e_2e_1^2-e_2e_1^3=0,\quad
e_2^2e_1-[2]_{q^2}e_2e_1e_2+e_1e_2^2=0,
\end{align*}
where $[m]_q=\frac{q^m-q^{-m}}{q-q^{-1}}$.
According to the general theory \cite{L2},
the $U^+_q(C_2)$ admits two natural PBW bases
$\{E_{\bf 1}^{a,b,c,d}\mid (a,b,c,d) \in \Z_{\ge 0}^4\}$ and 
$\{E_{\bf 2}^{a,b,c,d}\mid (a,b,c,d) \in \Z_{\ge 0}^4\}$ 
corresponding to the reduced expressions $s_1s_2s_1s_2$ and 
$s_2s_1s_2s_1$ of the longest element of $W(C_2)$.
Explicitly they read
\begin{equation}\label{proot}
\begin{split}
&E^{a,b,c,d}_{\bf 1} = \frac{(b'_4)^a(b'_3)^b(b'_2)^c(b'_1)^d}
{[a]_q![b]_{q^2}![c]_q![d]_{q^2}!},\quad
E^{a,b,c,d}_{\bf 2} = \frac{(b_1)^a(b_2)^b(b_3)^c(b_4)^d}
{[a]_{q^2}![b]_q![c]_{q^2}![d]_q!},\\
&b_1=e_2,\; b_2=e_1e_2-q^2e_2e_1,\;
b_3=\frac{e_1b_2-b_2e_1}{[2]_q},\; b_4=e_1,\\
&b'_1=e_2,\; b'_2=e_2e_1-q^2e_1e_2,\;
b'_3=\frac{b'_2e_1-e_1b'_2}{[2]_q},\; b'_4=e_1,
\end{split}
\end{equation}
where $[m]_q!= [m]_q[m-1]_q\cdots [1]_q$.
As the special case $\mathfrak{g}=C_2$ of the result 
for general simple Lie algebra $\mathfrak{g}$ \cite{KOY}, 
one has
\begin{align}\label{eke}
E^{a,b,c,d}_{\bf 2} = \sum_{i,j,k,l}
K^{a, b, c, d}_{i, j, k, l}\,E^{l,k,j,i}_{\bf 1}.
\end{align}
This relation in turn characterizes the matrix element $K^{a, b, c, d}_{i, j, k, l}$ 
of the intertwiner $K$ also as
the transition coefficient of the PBW bases.
See \cite{Sa, T} for the recent development on this topic.
Note that the equality of the weight on the both sides implies 
$K^{a,b,c,d}_{i,j,k,l}=0$ unless
$(a+b+c,b+2c+d)=(i+j+k,j+2k+l)$.

\section{Polynomial formula}\label{sec:pf}

\subsection{$K$ in terms of $Q_{b,c}(x,y,z,w)$} 
In \cite[Th. 3.4]{KO1} the matrix element $K^{a,b,c,d}_{i,j,k,l}$
was expressed as a finite sum of the form 
\begin{align}\label{mess}
K^{a,b,c,d}_{i,j,k,l}=\delta^{a+b+c}_{i+j+k}\delta^{b+2c+d}_{j+2k+l}
\sum_{\alpha,\beta, \gamma,\lambda}(-1)^{r_0}q^\Phi
\binom{r_1,\ldots, r_6}{s_1,\ldots, s_{11}}_{\!\!q^2}
\binom{r_7}{s_{12},s_{13}}_{\!\!q^4}.
\end{align}
Here $r_i, s_i$ are at most linear and $\Phi$ is at most quadratic
in $\alpha,\beta, \gamma,\lambda$. 
The indices $a,b,c,d,i,j,k,l$ enter many places 
and their dependence is quite involved.  

Our aim here is to provide an alternative formula which is more structural.
It is expressed in terms of 
{\em polynomials} in the four variables that accommodate the 
eigenvalues of $\OK$ and $\ok$ in the four components of 
$|i \rangle \otimes |j \rangle \otimes |k \rangle \otimes |l \rangle
\in {\mathcal F}_{q^2}\otimes {\mathcal F}_{q}\otimes 
{\mathcal F}_{q^2}\otimes {\mathcal F}_{q}$.

\begin{theorem}\label{th:main1}
(i) There are a family of polynomials 
$\{Q_{b,c}(x,y,z,w)\mid (b,c) \in \Z^2_{\ge 0}\}$ 
in variables $x,y,z,w$ 
characterized by the $q$-difference equations E22--E55 in Appendix \ref{app:de}
and the condition $Q_{0,0}(x,y,z,w)=1$.
(ii) The following formulas are valid for the matrix elements of the 3D $K$:
\begin{align}
&K^{a,b,c,d}_{i,j,k,l} = \delta^{a+b+c}_{i+j+k}\delta^{b+2c+d}_{j+2k+l}
\frac{q^{\phi_K-\varphi_{b,c}}}{(q^2)_b(q^4)_c}
Q_{b,c}(q^{4i},q^{2j},q^{4k},q^{2l}),\label{Kq}\\
&\phi_K = (a-k)(d-j)+(b-l)(c-i)-2(b-j)(c-k),\label{phiQ}\\
&\varphi_{b,c} = 3b(b-1)+2c(3c-2)+8bc .\label{pbc}
\end{align}
\end{theorem} 
\begin{proof}
The intertwining relations (\ref{reK}) for 
$K$ are listed in $\langle 22 \rangle - \langle 55 \rangle$ 
in Appendix \ref{app:de}.
Substituting (\ref{Kq}) into them one finds that they are equivalent to 
E22--E55 as illustrated there along $\langle24\rangle$. 
The normalization condition $K^{0,0,0,0}_{0,0,0,0}=1$ 
also matches $Q_{0,0}(x,y,z,w)=1$.
Thus there is a unique family of functions 
$\{Q_{b,c}(x,y,z,w)\mid (b,c) \in \Z^2_{\ge 0}\}$ satisfying E22--E55 
due to the unique existence of the intertwiner $K$.
Combining E22 and E23 in Appendix \ref{app:de}, one can derive
\begin{align}
Q_{b,c}(x,y,z,w)&= 
w y (z-1) q^{4 b+8 c-4}Q_{b-1,c}(x,y,q^{-4}z,w)\nonumber\\
&+w x (y-1) y z q^{4b+4c-4}Q_{b-1,c}(x,q^{-2}y,z,w)\nonumber\\
&+(w-1) (y-1) q^{6 b+8 c-6}Q_{b-1,c}(x,q^{-2}y,z,q^{-2}w)\nonumber\\
&+w (x-1) y (z-1) q^{4 b+8c-4}Q_{b-1,c}(q^{-4}x,q^2y,q^{-4}z,w)\nonumber\\
&+(w-1) (x-1) yq^{6 b+8 c-6}Q_{b-1,c}(q^{-4}x,y,z,q^{-2}w), \label{rec1}\\
Q_{b,c}(x,y,z,w)&= 
-w^2 y (z-1) z q^{4 b+8 c-8} Q_{b,c-1}(x,y,q^{-4}z,w)\nonumber\\
&+w x (y-1) z q^{4 b+4 c-6}(q^{2 (b+2c)}-q^2 w y z) Q_{b,c-1}(x,q^{-2}y,z,w)\nonumber\\
&-(w-1) w (y-1) z q^{6 b+8 c-8}Q_{b,c-1}(x,q^{-2}y,z,q^{-2}w)\nonumber\\
&+w(x-1) (z-1) q^{4 b+8 c-10}(q^{2 (b+2 c)}-q^2 w yz) 
Q_{b,c-1}(q^{-4}x,q^2 y,q^{-4}z,w)\nonumber\\
&+(w-1) (x-1) q^{6 b+8 c-10}(q^{2 (b+2c)}-q^2 w y z)
Q_{b,c-1}(q^{-4}x,y,z,q^{-2}w).\label{rec2}
\end{align}
By induction on $b$ and $c$ 
they tell that 
$Q_{b,c}(x,y,z,w)$ is a polynomial in $x,y,z,w$.
\end{proof}

The power $\phi_K$ in (\ref{phiQ}) is invariant under the 
exchange $(a,b,c,d) \leftrightarrow (i,j,k,l)$.
Therefore (\ref{kud}) implies another general formula:
\begin{align}
K^{a,b,c,d}_{i,j,k,l} = \delta^{a+b+c}_{i+j+k}\delta^{b+2c+d}_{j+2k+l}\;
q^{\phi_K-\varphi_{j,k}}
\binom{l}{b,d}_{\!\!q^2}\binom{i}{a,c}_{\!\!q^4}
Q_{j,k}(q^{4a},q^{2b},q^{4c},q^{2d}).\label{Kqq}
\end{align}

If one switches from $q$ to $p=q^{-1}$ and introduces the functions of 
$p,x,y,z,w$ by 
\begin{align*}
\hat{Q}_{b,c}(x,y,z,w) = p^{\varphi_{b,c}}Q_{b,c}(x,y,z,w)|_{q\rightarrow p^{-1}},
\end{align*}
the recursion relations (\ref{rec1}) and (\ref{rec2}) slightly simplify as
\begin{align}
\hat{Q}_{b,c}(x,y,z,w)
&=w y (z-1) p^{2 b-2} \hat{Q}_{b-1,c}(x,y,p^4 z,w)\nonumber\\
&+w x (y-1) y z p^{2 b+4c-2} \hat{Q}_{b-1,c}(x,p^2 y,z,w)\nonumber\\
&+(w-1) (y-1) \hat{Q}_{b-1,c}(x,p^2 y,z,p^2 w)\nonumber\\
&+w (x-1) y (z-1) p^{2 b-2}\hat{Q}_{b-1,c}(p^4 x,p^{-2}y,p^4 z,w)\nonumber\\
&+(w-1) (x-1) y \hat{Q}_{b-1,c}(p^4 x,y,z,p^2 w),\label{rec3}\\
\hat{Q}_{b,c}(x,y,z,w) 
&= -w^2 y (z-1) z p^{4 b+4 c-2} \hat{Q}_{b,c-1}(x,y,p^4 z,w)\nonumber\\
&-w x (y-1) z p^{2 b+4 c-6} (w y z p^{2 (b+2 c)}-p^2)
\hat{Q}_{b,c-1}(x,p^2 y,z,w)\nonumber\\
&-(w-1) w (y-1) z p^{2 b+4 c-2} \hat{Q}_{b,c-1}(x,p^2 y,z,p^2 w)\nonumber\\
&-w (x-1) (z-1) p^{2 b-2}(w y z p^{2 (b+2 c)}-p^2) \hat{Q}_{b,c-1}(p^4 x,p^{-2}y,p^4 z,w)\nonumber\\
&+(w-1) (x-1) (1-w y z p^{2 (b+2c-1)}) \hat{Q}_{b,c-1}(p^4 x,y,z,p^2 w).\label{rec4}
\end{align}
Here are some examples of 
$Q_{b,c}(x,y,z,w)$'s with small $b,c$:
\begin{align*}
Q_{1,0}(x,y,z,w)&= w x y^2 z-w-x y+1,\\
Q_{0,1}(x,y,z,w)&= q^2 (w x y z-w z-x+1)-w z \left(w x y^2 z-w-x y+1\right),\\
Q_{2,0}(x,y,z,w)&= q^6 (w-1) (x y-1)+q^4 \left(-w^2xy^2 z+w^2+w x y-w+x y^2-x y\right)\\
&-q^2 x y^2 (w x y z-w z-x+1)+w x y^2 z \left(w
   x y^2 z-w-x y+1\right)\\
&-q^4(w-w^2+xy-xyw-xy^2+xy^2zw^2),\\
Q_{1,1}(x,y,z,w)&=q^{10} (w-1) (x-1)-q^8 (w-1) w z (x y-1)\\
&+q^6 \left(-w^2 xy z+w^2 z-w x^2 y^2 z+2 w x y z-w z+x^2 y-x y\right)\\
&+q^4 wz \left(w^2 x y^2 z-w^2+w x^2 y^3 z-w x y^2 z-w x y+w-x^2 y^2+x y\right)\\
&+q^2 w x y^2 z (w x y z-w z-x+1)-w^2 x y^2z^2 \left(w x y^2 z-w-x y+1\right).
\end{align*}
One notices that these $Q_{b,c}(x,y,z,w)$ are polynomials also in $q^2$.
To show it in general we introduce 
\begin{align}\label{rstu}
{\mathscr S}_{b,c}=\{(r,s,t,u) \in \Z_{\ge 0}^4 
\mid \min(u-t, 2r-s, b-s+2t-u, c-r+s-t)\ge 0\},
\end{align}
which is a finite subset of 
$\{(r,s,t,u) \in \Z_{\ge 0}^4 \mid s/2\le r \le b+c,\, t\le u \le b+2c\}$.
\begin{proposition}\label{pr:tec}
(i) Let $Q_{b,c}(x,y,z,w) = \sum_{r,s,t,u} D^{b,c}_{r,s,t,u}x^ry^sz^tw^u$
with $D^{b,c}_{r,s,t,u}$ independent of $x,y,z,w$.
Then $D^{b,c}_{r,s,t,u} = 0$ unless $(r,s,t,u) \in {\mathscr S}_{b,c}$. 
(ii) $Q_{b,c}(x,y,z,w) \in \Z[q^2,x,y,z,w]$.
(iii) $q^{-\varphi_{b,c}}Q_{b,c}(x,y,z,w) \in \Z[q^{-2},x,y,z,w]$.
\end{proposition}
\begin{proof}
(i) By induction on $b$ and $c$,  
it suffices to show that 
all the monomials $x^ry^sz^tw^u$ which survive possible 
cancellations in the right hand sides of (\ref{rec1}) and 
(\ref{rec2}) satisfy $(r,s,t,u) \in {\mathscr S}_{b,c}$ 
assuming $D^{b-1,c}_{r,s,t,u}=0$ unless $(r,s,t,u) \in {\mathscr S}_{b-1,c}$
and  
$D^{b,c-1}_{r,s,t,u}=0$ unless $(r,s,t,u) \in {\mathscr S}_{b,c-1}$,
respectively.
We illustrate the procedure for (\ref{rec2}). 
The treatment of (\ref{rec1})  is completely similar.
First consider the condition $u-t\ge 0$ in (\ref{rstu}).
The right side of (\ref{rec2}) contains no prefactor $w^{u'}z^{t'}$ 
such that $u'-t'<0$, therefore this condition is trivially satisfied.
Second consider the condition $2r-s\ge 0$ in (\ref{rstu}).
The terms in the right side of (\ref{rec2}) that apparently break it are
\begin{align*}
&-w^2 y (z-1) z q^{4 b+8 c-8} Q_{b,c-1}(x,y,q^{-4}z,w)|_{2r=s}\\
&-(w-1) w yz q^{6 b+8 c-8}Q_{b,c-1}(x,q^{-2}y,z,q^{-2}w)|_{2r=s}\\
&-w(z-1) q^{4 b+8 c-10}(-q^2 w yz) 
Q_{b,c-1}(q^{-4}x,q^2 y,q^{-4}z,w)|_{2r=s}\\
&-(w-1)q^{6 b+8 c-10}(-q^2 w y z)
Q_{b,c-1}(q^{-4}x,y,z,q^{-2}w)|_{2r=s},
\end{align*}
where $|_{2r=s}$ means the contribution of the monomials 
$x^ry^sz^tw^u$ such that $2r=s$.
This vanishes due to 
$Q_{b,c-1}(x,y,q^{-4}z,w)|_{2r=s}
= Q_{b,c-1}(q^{-4}x,q^2 y,q^{-4}z,w)|_{2r=s}$
and 
$Q_{b,c-1}(x,q^{-2}y,z,q^{-2}w)|_{2r=s}
=Q_{b,c-1}(q^{-4}x,y,z,q^{-2}w)|_{2r=s}$.
Similarly the proof of the third and the fourth conditions 
$b-s+2t-u\ge 0$ and $c-r+s-t\ge 0$ in (\ref{rstu}) reduce to checking
\begin{align*}
&\phantom{+}w^2 yz q^{4 b+8 c-8} Q_{b,c-1}(x,y,q^{-4}z,w)|_{b=s-2t+u}\\
&-w^2yz q^{6 b+8 c-8}Q_{b,c-1}(x,q^{-2}y,z,q^{-2}w)|_{b=s-2t+u}\\
&-w(x-1)q^{4 b+8 c-10}(q^{2 (b+2 c)}-q^2 w yz) 
Q_{b,c-1}(q^{-4}x,q^2 y,q^{-4}z,w)|_{b=s-2t+u}\\
&+w(x-1)q^{6 b+8 c-10}(q^{2 (b+2c)}-q^2 w y z)
Q_{b,c-1}(q^{-4}x,y,z,q^{-2}w)|_{b=s-2t+u}=0,\\
&-wxz q^{4 b+4 c-6}(q^{2 (b+2c)}-q^2 w y z)
Q_{b,c-1}(x,q^{-2}y,z,w)|_{c-1=r-s+t}\\
&+wxz q^{4 b+8 c-10}(q^{2 (b+2 c)}-q^2 w yz) 
Q_{b,c-1}(q^{-4}x,q^2 y,q^{-4}z,w)|_{c-1=r-s+t}=0.
\end{align*}
Likewise the previous case, it is easy to see that these terms pairwise cancel.

(ii) By induction on $b$ and $c$, 
it suffices to show $Q_{b,c}(x,y,z,w) \in \Z[q^2,x,y,z,w]$ 
from (\ref{rec1}) and (\ref{rec2}) by 
assuming $Q_{b-1,c}(x,y,z,w) \in \Z[q^2,x,y,z,w]$
and  $Q_{b,c-1}(x,y,z,w) \in \Z[q^2,x,y,z,w]$, respectively.
This can easily be checked term by term 
in the right hand sides. 
For instance consider the contribution in (\ref{rec1}) whose $q$-dependent part is
$q^{4 b+8c-4}Q_{b-1,c}(q^{-4}x,q^2y,q^{-4}z,w)$.
It consists of the monomials of the form 
$q^{4 b+8c-4-4r+2s-4t}D^{b-1,c}_{r,s,t,u}x^ry^sz^tw^u$ with
$D^{b-1,c}_{r,s,t,u} \in \Z[q^2]$ by the assumption.
From (i) we know $D^{b-1,c}_{r,s,t,u}=0$ unless 
$(r,s,t,u) \in {\mathscr S}_{b-1,c}$. 
This ensures $4 b+8c-4-4r+2s-4t\in 2\Z_{\ge 0}$ if 
$D^{b-1,c}_{r,s,t,u}\neq 0$.

(iii) Likewise (ii) one can show $\hat{Q}_{b,c}(x,y,z,w) \in \Z[p^2,x,y,z,w]$ 
using (\ref{rec3}) and (\ref{rec4}).
\end{proof}

By a similar inductive argument it is easy to show 
\begin{proposition}\label{pr:m}
\begin{align*}
&\lim_{q\rightarrow 0} Q_{b,c}(x,y,z,w) 
= (-1)^c(xy^2)^{b+c-1}(zw)^{b+2c-1}Q_{1,0}(x,y,z,w),\\
&\lim_{q\rightarrow \infty} q^{-\varphi_{b,c}}Q_{b,c}(x,y,z,w)
= 1-xy^{\delta_c^0}-wz^{\delta_b^0}+xw y^{\delta_c^0+\delta_{b+c}^1}
z^{\delta_b^0+\delta_b^1\delta_c^0},\\
&Q_{b,c}(x,1,1,w) = (-1)^cx^{b+c}w^{b+2c}(x^{-1}; q^4)_{b+c}
(w^{-1}; q^2)_{b+2c},\\
&Q_{b,c}(x,y,1,1) = (-1)^c (xy^2)^{b+c}(y^{-1};q^2)_{b+2c},\\
&Q_{b,c}(1,1,z,w)= (-1)^c(zw)^{b+2c}(z^{-1};q^4)_{b+c},
\end{align*}
where $(b,c) \neq (0,0)$ in the first two relations.
\end{proposition}
Thus the power $\varphi_{b,c}$ (\ref{pbc}) gives  
the {\em exact} degree of $Q_{b,c}(x,y,z,w)$ in $q$. 
The third result for instance reflects 
$K^{a,b,c,d}_{i,0,0,l} = 0$ 
unless $(a,d) = (i-b-c,l-b-2c) \in \Z_{\ge 0}^2$.

Now we present an explicit form of $Q_{b,c}(x,y,z,w)$.
\begin{theorem}\label{th:main2}
The following formula is valid:
\begin{align}
&Q_{b,c}(x,y,z,w) =q^{\varphi_{b,c}}
\sum_{(r,s,t,u)\in {\mathscr S}_{b,c}}(-1)^{r+u}q^{\phi_Q-\psi_{r,s}}
C^{b,c}_{r,s,t,u}x^ry^sz^tw^u, \label{qc}\\
&C^{b,c}_{r,s,t,u} = 
\frac{\binom{b,u-t}{b+2t-s-u,2r-s}_{q^2}}
{(q^4)_r(q^4)_{u-t}(q^4)_{c-r+s-t}}(-1)^s q^{\psi_{r,s}}
\sum_{(\alpha,\beta,\gamma) \in \Z^3_{\ge 0}}(-1)^{\beta+\gamma}q^{\phi_{C}}
\,\Xi^{}_{\alpha,\beta,\gamma},\label{cxi}\\
&\Xi^{}_{\alpha,\beta,\gamma} =
\binom{b-s+t-\alpha,2r-s+\beta}{\alpha,\beta,\gamma,u-t-\alpha,t-\beta,
b-s-\alpha+\beta,s-\beta-\gamma}_{\!\!q^2}
\binom{c+s-r-\beta,c+\gamma}{c-r+\gamma}_{\!\!q^4},\label{Xi}\\
&\phi_Q = (s-2t+u)^2+2r(r+2t+1)-(2b-1)(s+u)-4c(r+t),\\
&\phi_C = \alpha(\alpha+1+2t)+\beta(\beta-1-2\alpha+2b-4r)
+\gamma(\gamma-1-4r),\\
&\psi_{r,s} = s(4r-s+1). \label{psirs}
\end{align}
\end{theorem}
In view of the support property of the symbols in (\ref{cbdef}),
the sum (\ref{cxi}) is limited to those $\alpha,\beta,\gamma$
such that all the lower entries in (\ref{Xi}) are nonnegative,
which also ensures that all the upper entries are so.
Thus it ranges over a finite subset of 
$\{(\alpha,\beta,\gamma) \in \Z_{\ge 0}^3\mid 
\alpha \le u-t,\; \beta \le t, \; \gamma \le s\}$. 
The dependence of 
$\Xi_{\alpha,\beta,\gamma}$ on $(b,c,r,s,t,u)$
has been suppressed in the notation for simplicity.

From Proposition \ref{pr:tec} (ii), we know 
$q^{\varphi_{b,c}+\phi_Q-\psi_{r,s}}C^{b,c}_{r,s,t,u} \in \Z[q^2]$.
Actually computer experiments suggest the following conjecture 
\begin{align}
C^{b,c}_{r,s,t,u} \in \Z[q^2],\quad
\lim_{q \rightarrow 0} C^{b,c}_{r,s,t,u} = 1.
\end{align}
In general $\Xi_{\alpha,\beta,\gamma}$ is not necessarily a polynomial 
but a rational function of $q^2$.
In the special case $b=0$ or $c=0$, $C^{b,c}_{r,s,t,u}$ admits a simple formula:
\begin{align}
C^{b,0}_{r,s,t,u}&= \binom{b}{u,b+2t-s-u,2r-s}_{\!\!q^2}
\binom{u}{t,u-t,s-r-t}_{\!\!q^4},\\
C^{0,c}_{r,s,t,u}&=\binom{2r}{s,2t-s-u,2r-s}_{\!\!q^2}
\binom{c}{r,u-t,c-r+s-t}_{\!\!q^4}.
\end{align}
The latter is equivalent to 
$A^{0,c}_{r,s,t,u}= (q^2)_t(q^2)_{2r}/(q^4)_r$ in (\ref{CA}), 
which can be verified by following Steps (i)--(v) in Section \ref{app:sol}.

\begin{example}\label{ex:k}
The following is the list of the nonzero $K^{3,1,0,2}_{i,j,k,l}$:
\begin{align*}
K^{3,1,0,2}_{1,3,0,0}&= -q^6(1-q+q^2)(1+q+q^2),\\
K^{3,1,0,2}_{2,1,1,0}&= -q^{10}(1-q+q^2)(1+q+q^2),\\
K^{3,1,0,2}_{2,2,0,1}&= (1+q^2)(1 - q^2 + q^4 - q^6 + q^8 - q^{10} - q^{14}),\\
K^{3,1,0,2}_{3,0,1,1}&= q^6(1+q^2)(1 - q^2 + q^4 - q^6 + q^8 - q^{10} - q^{14}),\\
K^{3,1,0,2}_{3,1,0,2}&= q^6(1 + q^2 - q^{14} - q^{16} - q^{18}),\\
K^{3,1,0,2}_{4,0,0,3}&=q^{14}(1-q+q^2)(1+q+q^2)(1-q^{16}).
\end{align*}
According to (\ref{Kq}) they are expressed by various special values of  
$Q_{1,0}(x,y,z,w)= w x y^2 z-w-x y+1$ as
\begin{align*}
K^{3,1,0,2}_{1,3,0,0}&= \frac{Q_{1,0}(q^4,q^6,1,1)}{q^4(1-q^2)},\;
K^{3,1,0,2}_{2,1,1,0}=\frac{Q_{1,0}(q^8,q^2,q^4,1)}{1-q^2},\;
K^{3,1,0,2}_{2,2,0,1}=\frac{Q_{1,0}(q^8,q^4,1,q^2)}{1-q^2},\\
K^{3,1,0,2}_{3,0,1,1}&= \frac{q^6 Q_{1,0}(q^{12},1,q^4,q^2)}{1-q^2},\;
K^{3,1,0,2}_{3,1,0,2}= \frac{q^6 Q_{1,0}(q^{12},q^2,1,q^4)}{1-q^2},\;
K^{3,1,0,2}_{4,0,0,3}=\frac{q^{14} Q_{1,0}(q^{16},1,1,q^6)}{1-q^2}.
\end{align*}
On the other hand according to (\ref{Kqq}) they are also  
expressed in terms of $Q_{b,c}(q^{12},q^2,1,q^4)$ with various $b,c$.
For instance one has
\begin{align*}
K^{3,1,0,2}_{2,1,1,0}=
\frac{q^{-10}Q_{1,1}(q^{12},q^2,1,q^4)}{(1-q^2)^2(1-q^4)(1-q^{12})},
\quad
K^{3,1,0,2}_{3,0,1,1}= 
\frac{q^4 Q_{0,1}(q^{12},q^2,1,q^4)}{(1-q^2)(1-q^4)}.
\end{align*} 
\end{example}

\subsection{Proof of Theorem \ref{th:main2}}\label{app:sol}

We find it convenient to introduce $A^{b,c}_{r,s,t,u}$ by 
\begin{align}\label{CA}
C^{b,c}_{r,s,t,u} = \binom{b}{s,t,b-s+2t-u,2r-s}_{\!\!q^2}
\binom{c}{u-t,c-r+s-t}_{\!\!q^4}A^{b,c}_{r,s,t,u}.
\end{align}
Then substitution of (\ref{qc}) into the difference equations 
E22--E55 in 
Appendix \ref{app:de} leads to the recursion relations for $A^{b,c}_{r,s,t,u}$.
In what follows we outline how they can be solved in Steps (i)--(v) to yield 
(\ref{qc})--(\ref{psirs}).
An important feature in the derivation is the boundary condition 
$A^{b,c}_{r,s,t,u}=0$ unless $(r,s,t,u) \in {\mathscr S}_{b,c}$
which stems from Proposition \ref{pr:tec} (i).

(i) Combining E35 and E45, one gets
\begin{align*}
A^{b,c}_{r,s,t,u} = A^{b,c}_{r,s,t,u-1}+q^{2u}A^{b-1,c}_{r,s,t,u-1}\quad (t<u),
\end{align*}
which leads to 
\begin{align*}
A^{b,c}_{r,s,t,u} = \sum_{\alpha=0}^{u-t}q^{\alpha(\alpha+2t+1)}
\binom{u-t}{\alpha,u-t-\alpha}_{\!\!q^2}A^{b-\alpha,c}_{r,s,t,t}.
\end{align*}
Henceforth we concentrate on $A^{b,c}_{r,s,t,t}$ with 
$\min(b-s+t,c+s-r-t,2r-s) \ge 0$.

(ii) From E24, one gets
\begin{align*}
A^{b,c}_{r,s,t,t} = q^{2(b-s+t)}(1-q^{2s})A^{b,c}_{r,s-1,t-1,t-1}+
(1-q^{2(b-s+t)})A^{b,c}_{r,s,t-1,t-1}\quad
(b-s+t\ge 0),
\end{align*}
which leads to 
\begin{align*}
A^{b,c}_{r,s,t,t}=\sum_{\beta=(t-s)_+}^{t-(s-b)_+}
q^{2(t-\beta)(b-s+t-\beta)}
\binom{s,t,b-s+t}{\beta,t-\beta,s-t+\beta,b-s+t-\beta}_{\!\!q^2}
A^{b,c}_{r,s-t+\beta,0,0},
\end{align*}
where $(x)_+ = \max(x,0)$.
Under the assumption 
$\min(b-s+t,c+s-r-t,2r-s) \ge 0$, 
all the summands $A^{b,c}_{r,m, 0,0}$ appearing here 
satisfy the condition $(r,m,0,0) \in {\mathscr S}_{b,c}$ in (\ref{rstu}). 
Henceforth we concentrate on $A^{b,c}_{r,s,0,0}$ with 
$\min(b-s,c+s-r,2r-s)\ge 0$.

(iii) By setting $t=u=0$ in the recursion relation of 
$A^{b,c}_{r,s,t,u}$ derived from E34, one gets
\begin{align}\label{AA}
A^{b,c}_{r,s,0,0} = A^{b-1,c}_{r,s,0,0}\quad (b>s).
\end{align}
Henceforth we concentrate on $A^{b,c}_{r,b,0,0}$ with 
$\min(b+c-r,2r-b)\ge 0$.

(iv) By setting $t=u=0$ in the recursion relation of 
$A^{b,c}_{r,s,t,u}$ derived from E32, one gets
\begin{align*}
(1-q^{4+4c})A^{b-1,c+1}_{r,s-1,0,0}
-q^{2+4r-2s}(1-q^{4(c-r+s)})A^{b,c}_{r,s-1,0,0}
-(1-q^{2(2r-s+1)})A^{b,c}_{r,s,0,0}=0.
\end{align*}
Setting $s=b$ and applying (\ref{AA}) further, one finds
\begin{align*}
(1-q^{4+4c})A^{b-1,c+1}_{r,b-1,0,0}
-q^{2+4r-2b}(1-q^{4(b+c-r)})A^{b-1,c}_{r,b-1,0,0}
-(1-q^{2(2r-b+1)})A^{b,c}_{r,b,0,0}=0
\end{align*}
for $2r\ge b,\; b+c\ge r$. 
This allows one to decrease $b$, leading to
\begin{align*}
A^{b,c}_{r,b,0,0} = \sum_{\gamma=0}^{b-(r-c)_+}
(-1)^\gamma q^{\gamma(\gamma+1+4r-2b)}
\binom{b,2r-b}{2r,\gamma,b-\gamma}_{\!\!q^2}
\binom{b+c-r, b+c-\gamma}{c, b+c-r-\gamma}_{\!\!q^4}
A^{0,b+c-\gamma}_{r,0,0,0}.
\end{align*}
Henceforth we concentrate on $A^{0,c}_{r,0,0,0}$ with $c\ge r$.

(v) By setting $b=s=t=u=0$ in the recursion relation of 
$A^{b,c}_{r,s,t,u}$ derived from E53, one gets
$A^{0,c}_{r,0,0,0}=(1-q^{4r-2})A^{0,c}_{r-1,0,0,0}$.
Thus $A^{0,c}_{r,0,0,0} = \frac{(q^2)_{2r}}{(q^4)_r}A^{0,c}_{0,0,0,0}$.
By setting $b=r=s=t=u=0$ in the recursion relation of 
$A^{b,c}_{r,s,t,u}$ derived from E23 or E54, one gets
$A^{0,c}_{0,0,0,0}=A^{0,c-1}_{0,0,0,0}$.
From $Q_{0,0}(x,y,w,z)=1$ it follows that $A^{0,0}_{0,0,0,0}=1$.
Therefore $A^{0,c}_{r,0,0,0} = \frac{(q^2)_{2r}}{(q^4)_r}\; (c\ge r)$.

Synthesizing the results in Steps (i)--(v)  one obtains a formula for
$A^{b,c}_{r,s,t,u}$ as a triple sum with respect to $(\alpha, \beta, \gamma)$
over a finite subset of $\Z_{\ge 0}^3$.
Express the ratio of the 
$q$-factorials contained in the summand by using the symbol (\ref{cbdef}). 
Then a little inspection shows that the sum can actually be relaxed to 
$(\alpha, \beta, \gamma) \in \Z_{\ge 0}^3$ due to its support property
mentioned after (\ref{cbdef}).
Together with (\ref{CA}) one arrives at (\ref{cxi})--(\ref{psirs}).
This completes the proof of Theorem \ref{th:main2}.

\section{Analogous result on 3D $R$}\label{sec:R}
Our result may be viewed as a generalization of the analogous 
fact for the 3D $R$ \cite{BS,BMS}. 
We review it here with a few supplements for comparison, 
as it is a closely related object associated with the 
quantized algebra of functions $A_q(A_2)$ \cite{KV,KO1} and 
constitute the companion in the 3D reflection equation (\ref{a:3dref}).
As an application the expression of the 3D $R$ (\ref{ko})
announced in \cite[eq. (9)]{KO3} 
is derived here for the first time.
See \cite{BS,BMS,KO1} for more aspects.

The 3D $R$ is the linear operator 
on ${\mathcal F}_q^{\otimes 3}$ whose 
parameter (except $q$) free part is characterized,  
up to a normalization, by 
\begin{equation}\label{recR}
\begin{split}
&R(\apm \otimes \ok \otimes \ichi) = 
(\apm \otimes  \ichi \otimes \ok + \ok \otimes \apm \otimes \amp )R,\\
&R(\ichi \otimes \ok \otimes \apm) = 
(\ok \otimes  \ichi \otimes \apm + \amp \otimes \apm \otimes \ok)R,\\
&R(\ichi \otimes \apm \otimes \ichi) = 
(\apm \otimes \ichi \otimes \apm - q\,\ok \otimes \apm \otimes \ok)R,\\
&[R,\ok \otimes \ok \otimes \ichi] = [R, \ichi\otimes \ok \otimes \ok] = 0.
\end{split}
\end{equation}
They follow either as the intertwining relations for the 
irreducible $A_q(A_2)$ modules \cite{KV} 
or from a quantum geometry consideration \cite{BS,BMS}.
Actually there should be {\em nine} intertwining relations in total
reflecting the $3 \times 3$ matrix nature of $A_q(A_2)$.
The last one is given by 
\begin{align}\label{t22}
R(\ap\otimes \am\otimes \ap-q\,\ok \otimes \ichi \otimes \ok)
=(\am\otimes \ap\otimes \am-q\,\ok \otimes \ichi \otimes \ok)R.
\end{align}
Define the matrix elements of the $R$ by 
$R(|i\rangle \otimes |j\rangle \otimes |k\rangle) = \sum_{a,b,c}R^{a,b,c}_{i,j,k} 
|a\rangle \otimes |b\rangle \otimes |c\rangle$.
We adopt the normalization $R^{0,0,0}_{0,0,0}=1$.
Setting 
$R^{a,b,c}_{i,j,k} = \delta^{a+b}_{i+j}\delta^{b+c}_{j+k}\;
q^{(a-j)(c-j)}P_b(q^{2i},q^{2j},q^{2k})/(q^2)_b$ \cite{BS},
the relations 
(\ref{recR}) and (\ref{t22}) 
are translated into the $q$-difference equations:
\begin{align}
&P_b(q^2x,y,z)-P_b(x,y,z)
-q^{2-2b}x(1-q^{2b})(1-q^{2-2b}yz)P_{b-1}(x,y,z)=0,\label{42}\\
&P_b(x,y,q^2z)-P_b(x,y,z)
-q^{2-2b}z(1-q^{2b})(1-q^{2-2b}xy)P_{b-1}(x,y,z)=0,\\
&(1-x)P_b(q^{-2}x,y,z)-q^{-2b}z(1-q^{-2b}xy)P_b(x,y,z)
-P_{b+1}(x,y,z)=0,\label{44}\\
&(1-z)P_b(x,y,q^{-2}z)-q^{-2b}x(1-q^{-2b}yz)P_b(x,y,z)
-P_{b+1}(x,y,z)=0,\label{45}\\
&P_b(x,q^2y,z)-P_b(x,y,z)+q^{4-4b}xyz(1-q^{2b})P_{b-1}(x,y,z)=0,\\
&yP_{b+1}(x,y,z)+(1-y)P_b(x,q^{-2}y,z)
-(1-q^{-2b}xy)(1-q^{-2b}yz)P_b(x,y,z)=0,\\
&(y-q^{2b})P_b(x,y,z)+(1-y)P_b(q^2x,q^{-2}y,q^2z)
=(1-q^{2b})(1-q^{2-2b}xy)(1-q^{2-2b}yz)P_{b-1}(x,y,z)
\end{align}
and the condition $P_0(x,y,z)=1$.
It is known that $R=R^{-1}$ \cite[Prop. 2.4]{KO1}.
Applying this to (\ref{recR}) and (\ref{t22}) one can extract another set of
$q$-difference equations:
\begin{align}
&P_b(x,y,z)-q^{-2b}zP_b(q^2x,y,z)-(1-z)P_b(x,q^2y,q^{-2}z)=0,\\
&P_b(x,y,z)-q^{-2b}xP_b(x,y,q^2z)-(1-x)P_b(q^{-2}x,q^2y,z)=0,\\
&q^{-2b}x(1-y)P_b(x,q^{-2}y,q^2z)+(1-x)P_b(q^{-2}x,y,z)
-(1-q^{-2b}xy)P_b(x,y,z)=0,\\
&q^{-2b}z(1-y)P_b(q^2x,q^{-2}y,z)+(1-z)P_b(x,y,q^{-2}z)
-(1-q^{-2b}yz)P_b(x,y,z)=0,\\
&(1-q^{2b})P_{b-1}(x,y,z)-P_b(q^2x,y,q^2z)+q^{2b}P_b(x,q^2y,z)=0,\\
&q^{-2b}xz(1-y)P_b(x,q^{-2}y,z)-(1-x)(1-z)P_b(q^{-2}x,y,q^{-2}z)
+P_{b+1}(x,y,z)=0,\label{BS}\\
&P_{b+1}(x,y,z)-(1-x)(1-z)P_b(q^{-2}x,q^2y,q^{-2}z)
-xzq^{-4b}(y-q^{2b})P_b(x,y,z)=0.\label{p22}
\end{align}
The recursion (\ref{BS}), which was adopted as the defining relation of 
$P_b(x,y,z)$ in \cite{BS}, 
is a member of these compatible system of difference equations.
Any one of the recursions (\ref{44}), (\ref{45}) or (\ref{p22})
also determines $P_b(x,y,z)$ uniquely as a polynomial in $x,y,z$.
(More precisely $P_b(x,y,z) \in q^{-2b(b-1)}\Z[q^2,x,y,z]$ holds.)
The original problem (\ref{recR}) is symmetric 
with respect to the interchange of the first and third components.
Accordingly $P_b(x,y,z)=P_b(z,y,x)$ holds and there are 
four pairs of relations in (\ref{42})--(\ref{p22}) connected by this symmetry.
The solution admits an explicit formula \cite{BMS}
\begin{align}\label{qhg}
P_b(x,y,z) &= (q^{2-2b}z;q^2)_b\,
{}_2\phi_1(q^{-2b},q^{2-2b}yz,q^{2-2b}z; q^2,q^2x)
\end{align}
in terms of the $q$-hypergeometric series
${}_2\phi_1(a,b,c; q, z) = \sum_{n\ge 0}
\frac{(a; q)_n(b; q)_n}{(q; q)_n(c; q)_n}z^n,$
which is actually terminating in (\ref{qhg}) due to the entry $q^{-2b}$.
It is most easily established from (\ref{45}) by gathering the terms
in powers of $x$.

Another interesting result is the integral formula \cite[eq. (104)]{S}:
\begin{align}\label{if}
P_b(x,y,z) = q^{-b(b-1)}(q^2)_b \oint\frac{du}{2\pi {\mathrm i}u^{b+1}}
\frac{(-q^{2-2b}xyzu;q^2)_\infty(-u;q^2)_\infty}
{(-xu;q^2)_\infty(-zu;q^2)_\infty},
\end{align}
where the integral (\ref{if}) encircles $u=0$ anti-clockwise 
picking up the residue.
Equivalently the generating series is factorized as
\begin{align}\label{gs}
\sum_{b \ge 0} \frac{q^{b(b-1)}u^b}{(q^2)_b}
P_b(x,q^{2b-2}y,z) = \frac{(-xyzu;q^2)_\infty(-u;q^2)_\infty}
{(-xu;q^2)_\infty(-zu;q^2)_\infty}.
\end{align}
Substitution of (\ref{if}) 
into (\ref{42})--(\ref{p22}) gives rise to two situations.
In the simple case the integrands just sum up to zero 
under an appropriate rescaling of $u$.
The other case requires a slight maneuver.
Let us illustrate it along (\ref{BS}) as an example.
After the substitution of (\ref{if}) 
and replacement of $u$ by $q^2u$ in the second term,
one is left to show 
\begin{align*}
\oint
\frac{du(-q^{-2b}xyzu;q^2)_\infty(-q^2u;q^2)_\infty}
{u^{b+2}(-xu;q^2)_\infty(-zu;q^2)_\infty}
\left(xz(1-y)u(1+u)-(1-x)(1-z)u+(1-q^{2b+2})(1+u)\right)=0.
\end{align*}
By setting $f(u) =(-q^{-2-2b}xyzu;q^2)_\infty(-u;q^2)_\infty
/((-xu;q^2)_\infty(-zu;q^2)_\infty)$,
this is identified with the
identity $0=\oint \frac{du}{u^{b+1}}(f(q^2u)-q^{2b+2}f(u))$.
All the relations (\ref{42})--(\ref{p22}) can be verified in a similar manner.

Due to (\ref{if}) matrix elements of the 3D $R$ are expressed as \cite{BMS}
\begin{align*}
R^{a,b,c}_{i,j,k} = \delta^{a+b}_{i+j}\delta^{b+c}_{j+k}\,q^{ik+b}
\oint\frac{du}{2\pi {\mathrm i}u^{b+1}}
\frac{(-q^{2+a+c}u;q^2)_\infty(-q^{-i-k}u;q^2)_\infty}
{(-q^{a-c}u;q^2)_\infty(-q^{c-a}u;q^2)_\infty}.
\end{align*}
Note that the ratio of the infinite products equals 
$(-q^{-i-k}u;q^2)_i/(-q^{c-a}u;q^2)_{a+1}$ because of 
$a-c=i-k$.
Applying the standard expansion to it and collecting the coefficients of $u^b$,
one gets
\begin{align}\label{ko}
R^{a,b,c}_{i,j,k} = \delta^{a+b}_{i+j}\delta^{b+c}_{j+k}
\sum_{\lambda, \mu}(-1)^\lambda
q^{ik+b+\lambda(c-a)+\mu(\mu-i-k-1)}
\binom{\lambda+a}{\lambda, a}_{\!q^2}\binom{i}{\mu,i-\mu}_{\!q^2}
\end{align}
summed over $\lambda, \mu \in \Z_{\ge 0}$ under  the constraint 
$\lambda+\mu=b$. 
This was announced in \cite[eq. (9)]{KO3}. 

\section{Concluding remarks}\label{sec:sum}

In this paper we have introduced a family of polynomials 
$Q_{b,c}(x,y,z,w)\in \Z[q^2,x,y,z,w]$ characterized by the recursion relations 
(\ref{rec1}) and  (\ref{rec2}), 
or more generally E22--E55 in Appendix \ref{app:de}.
They form a compatible set of $q$-difference equations  
associated with the `maximal' representation $\pi_{1212} \simeq \pi_{2121}$ 
of the quantized algebra of functions $A_q(C_2)$.
The variables $x,y,z,w$ correspond to the four positive roots of $C_2$ 
as indicated in (\ref{proot}) and (\ref{eke}).
We have shown some basic properties of the polynomials in Proposition \ref{pr:tec}
and Proposition \ref{pr:m}.
The $q$-difference equations are solved in Section \ref{app:sol} and 
a new formula of the 3D $K$, the solution \cite{KO1}  
to the 3D reflection equation (\ref{a:3dref}),  
is obtained in Theorem \ref{th:main1} and Theorem \ref{th:main2}.
We have also included an expanded review on the 
closely related result on the 3D $R$ in Section \ref{sec:R}. 
It is an interesting question if the family of polynomials $Q_{b,c}(x,y,z,w)$
admit a factorizable generating series analogous to (\ref{gs}).
Another challenge will be to establish a similar polynomial formula 
for the intertwiner of $A_q(G_2)$ \cite[Sec.~4.4]{KOY} 
for which {\em no} general expression has been constructed. 

We remark that the system of intertwining relations 
$\langle 22 \rangle - \langle 55 \rangle$ in Appendix \ref{app:de} 
is {\em autonomous} in the sense that the apparent $q$ 
can completely be removed by replacing 
$\ok, \OK$ by
$q^{-1/2}\ok$ and 
$q^{-1}\OK$, respectively.
The new ones act on the Fock space by
$\ok|m\rangle = 
q^{m+\frac{1}{2}}|m\rangle,\,
\OK|m\rangle = 
q^{2m+1}|m\rangle$.
The resulting relations 
$\ok \,{{\rm{\bf a}}^{\!\pm} }
=q^{\pm 1}{{\rm{\bf a}}^{\!\pm} }\ok,\,
{{\rm{\bf a}}^{\!\pm} }{{\rm{\bf a}}^{\!\mp} }
= 1-q^{\mp1}\ok^2$
and 
$\OK {{\rm{\bf A}}^{\!\pm} }
=q^{\pm 2}{{\rm{\bf A}}^{\!\pm} }\OK,\,
{{\rm{\bf A}}^{\!\pm} }{{\rm{\bf A}}^{\!\mp} }
= 1-q^{\mp2}\OK^2$
can be realized in terms of the Weyl pairs
$\langle \ok, {\rm{\bf w}}\rangle$ and 
$\langle \OK, {\rm{\bf W}}\rangle$
satisfying  
$\ok {\rm{\bf w}} = q {\rm{\bf w}}\ok$ and
$\OK {\rm{\bf W}} = q^2 {\rm{\bf W}}\OK$ by
\begin{align*}
&\ap = (1-q^{-1}\ok^2)^{1/2} {\rm{\bf w}},\quad\;\;\;\,
\am = (1-q\ok^2)^{1/2} {\rm{\bf w}}^{-1},\\
&\Ap = (1-q^{-2}\OK^2)^{1/2} {\rm{\bf W}},\quad
\Am = (1-q^2\OK^2)^{1/2} {\rm{\bf W}}^{-1}.
\end{align*}
It is an interesting problem to seek a solution $K$ to the 
intertwining relations $\langle 22 \rangle - \langle 55 \rangle$ 
in Appendix \ref{app:de}
for the canonical representation of the Weyl pairs
(called non-compact representation of the $q$-oscillator algebra \cite{Sc}).
Especially the autonomous feature mentioned above 
indicates a possible extension to the modular double setting 
where not only $q$ but also its modular dual $\tilde{q}$ 
($(\log q)(\log \tilde{q}) = \mathrm{const}$) enters everywhere
compatibly via the Faddeev non-compact quantum dilogarithm \cite{BMS,KOS}.
Such an analysis effectively poses 
an analytic continuation of the eigenvalues of $\ok$ and $\OK$
away from $q^{\frac{1}{2}+\Z_{\ge 0}}$ and
$q^{1+2\Z_{\ge 0}}$.
The result in this paper may be viewed as a first step in this direction.

\section*{Acknowledgments}
The authors thank Sergey Sergeev for communication on reference.
This work is supported by 
Grants-in-Aid for Scientific Research No.~24540203
from JSPS.

\appendix
\section{Difference equations for $Q_{b,c}(x,y,z,w)$}\label{app:de}

Let $\langle\mathrm{ i j} \rangle$ be the 
intertwining relation for (\ref{reK}) $K$ (rather than ${\mathscr K}$) 
with the choice $f = t_{i-1,j-1}$.
They all become independent of the parameters 
in (\ref{pi31}) and (\ref{pi33}) other than $q$.
Explicitly they read as follows (cf. \cite[App.A]{KO1}):
\begin{alignat*}{2}
&\langle 2 2 \rangle: & \quad &
[ \ichi \!\otimes\! \am \!\otimes\! \ichi \!\otimes\! \am
- q \ichi \!\otimes\! \ok \!\otimes\! \Am \!\otimes\! \ok, \,K] = 0,\\
&\langle 2 3 \rangle: & \quad &
(\ichi \!\otimes\! \am \!\otimes\! \ichi \!\otimes\! \ok
+ \ichi \!\otimes\! \ok \!\otimes\! \Am \!\otimes\! \ap)K\\
&&& =
K(\Am \!\otimes\! \ap \!\otimes\! \Am \!\otimes\! \ok+
\Am \!\otimes\! \ok \!\otimes\! \ichi \!\otimes\! \am -q^2
\OK \!\otimes\! \am \!\otimes\! \OK \!\otimes\! \ok),\\
&\langle 2 4 \rangle: & \quad &
(\ichi \!\otimes\! \ok \!\otimes\! \OK \!\otimes\! \am) K = 
K(
\Ap \!\otimes\! \am \!\otimes\! \OK \!\otimes\! \ok + 
\OK \!\otimes\! \ap \!\otimes\! \Am \!\otimes\! \ok +
\OK \!\otimes\! \ok \!\otimes\! \ichi \!\otimes\! \am),\\
&\langle 2 5 \rangle: & \quad &
[\ichi \!\otimes\! \ok \!\otimes\! \OK \!\otimes\! \ok, K] = 0,\\
&\langle 3 2 \rangle: & \quad &
(\Am \!\otimes\! \ap \!\otimes\! \Am \!\otimes\! \ok
+\Am \!\otimes\! \ok \!\otimes\! \ichi \!\otimes\! \am -q^2
\OK \!\otimes\! \am \!\otimes\! \OK \!\otimes\! \ok)K\\
&&& =K( \ichi \!\otimes\! \am \!\otimes\! \ichi \!\otimes\! \ok +
\ichi \!\otimes\! \ok \!\otimes\! \Am \!\otimes\! \ap),\\
&\langle 3 3 \rangle: & \quad &
[\Am \!\otimes\! \ap \!\otimes\! \Am \!\otimes\! \ap -q
\Am \!\otimes\! \ok \!\otimes\! \ichi \!\otimes\! \ok -q^2
\OK \!\otimes\! \am \!\otimes\! \OK \!\otimes\! \ap, \,K] = 0,\\
&\langle 3 4 \rangle: & \quad &
(\Am \!\otimes\! \ap \!\otimes\! \OK \!\otimes\! \am + 
\OK \!\otimes\! \am \!\otimes\! \Ap \!\otimes\! \am - q
\OK \!\otimes\! \ok \!\otimes\! \ichi \!\otimes\! \ok)K\\
&&&=
K(\Ap \!\otimes\! \am \!\otimes\! \OK \!\otimes\! \ap+
\OK \!\otimes\! \ap \!\otimes\! \Am \!\otimes\! \ap -q
\OK \!\otimes\! \ok \!\otimes\! \ichi \!\otimes\! \ok),\\
&\langle 3 5 \rangle: & \quad &
(\Am \!\otimes\! \ap \!\otimes\! \OK \!\otimes\! \ok+
\OK \!\otimes\! \am \!\otimes\! \Ap \!\otimes\! \ok+
\OK \!\otimes\! \ok \!\otimes\! \ichi \!\otimes\! \ap )K=
K(
\ichi \!\otimes\! \ok \!\otimes\! \OK \!\otimes\! \ap),\\
&\langle 4 2 \rangle: & \quad &
(\Ap \!\otimes\!\am \!\otimes\! \OK \!\otimes\! \ok + 
\OK \!\otimes\! \ap \!\otimes\! \Am \!\otimes\! \ok +
\OK \!\otimes\! \ok \!\otimes\! \ichi \!\otimes\! \am)K
=
K(\ichi \!\otimes\! \ok \!\otimes\! \OK \!\otimes\! \am),\\
&\langle 4 3 \rangle: & \quad &
(\Ap \!\otimes\! \am \!\otimes\! \OK \!\otimes\! \ap+
\OK \!\otimes\! \ap \!\otimes\! \Am \!\otimes\! \ap - q
\OK \!\otimes\! \ok \!\otimes\! \ichi \!\otimes\! \ok)K\\
&&& =
K(\Am \!\otimes\! \ap \!\otimes\! \OK \!\otimes\! \am +
\OK \!\otimes\! \am \!\otimes\! \Ap \!\otimes\! \am -q
\OK \!\otimes\! \ok \!\otimes\! \ichi \!\otimes\! \ok),\\
&\langle 4 4 \rangle: & \quad &
[\Ap \!\otimes\! \am \!\otimes\! \Ap \!\otimes\! \am -q 
\Ap \!\otimes\! \ok \!\otimes\! \ichi \!\otimes\! \ok -q^2
\OK \!\otimes\! \ap \!\otimes\! \OK \!\otimes\! \am,\, K]=0,\\
&\langle 4 5 \rangle: & \quad &
(\Ap \!\otimes\! \am \!\otimes\! \Ap \!\otimes\! \ok +
\Ap \!\otimes\! \ok \!\otimes\! \ichi \!\otimes\! \ap -q^2
\OK \!\otimes\! \ap \!\otimes\! \OK \!\otimes\! \ok)K\\
&&&= K(\ichi \!\otimes\! \ap \!\otimes\! \ichi \!\otimes\! \ok +
\ichi \!\otimes\! \ok \!\otimes\! \Ap \!\otimes\! \am),\\
&\langle 5 2 \rangle: & \quad &
[\ichi \!\otimes\! \ok \!\otimes\! \OK \!\otimes\! \ok,\, K] = 0
\quad 
(\text{same as $\langle 2 5 \rangle$}),\\
&\langle 5 3 \rangle: & \quad &
(\ichi \!\otimes\! \ok \!\otimes\! \OK \!\otimes\! \ap)K = 
K(\Am \!\otimes\! \ap \!\otimes\! \OK \!\otimes\! \ok +
\OK \!\otimes\! \am \!\otimes\! \Ap \!\otimes\! \ok+
\OK \!\otimes\! \ok \!\otimes\! \ichi \!\otimes\! \ap),\\
&\langle 5 4 \rangle: & \quad &
(\ichi \!\otimes\! \ap \!\otimes\! \ichi \!\otimes\! \ok + 
\ichi \!\otimes\! \ok \!\otimes\! \Ap \!\otimes\! \am)K\\
&&&=K(\Ap \!\otimes\! \am \!\otimes\! \Ap \!\otimes\! \ok+
\Ap \!\otimes\! \ok \!\otimes\! \ichi \!\otimes\! \ap - q^2
\OK \!\otimes\! \ap \!\otimes\! \OK \!\otimes\! \ok),\\
&\langle 5 5 \rangle: & \quad &
[ \ichi \!\otimes\! \ap \!\otimes\! \ichi \!\otimes\! \ap -q
\ichi \!\otimes\! \ok \!\otimes\! \Ap \!\otimes\! \ok,\, K] = 0.
\end{alignat*}
The relations $\langle25\rangle$ and $\langle52\rangle$ 
imply the factor 
$\delta^{a+b+c}_{i+j+k}\delta^{b+2c+d}_{j+2k+l}$ in (\ref{Kq}).
The other ones are translated into difference equations of 
$Q_{b,c}(x,y,z,w)$.
For instance consider the equation $\langle 24 \rangle$.
In terms of the matrix elements defined by (\ref{Kact}) it reads
\begin{align*}
&q^{b+2c}(1-q^{2d+2})K^{a,b,c,d+1}_{i,j,k,l} \\
&= q^{2k+l}(1-q^{2j})K^{a,b,c,d}_{i+1,j-1,k,l} 
+ q^{2i+l}(1-q^{4k})K^{a,b,c,d}_{i,j+1,k-1,l}
+q^{2i+j}(1-q^{2l})K^{a,b,c,d}_{i,j,k,l-1}.
\end{align*}
Substituting (\ref{Kq}) into this and setting 
$(x,y,z,w) = (q^{4i},q^{2j},q^{4k},q^{2l})$ 
one gets the difference equation $\text{E}24$ given below.
Similarly, the equation $\langle ij \rangle$ 
is cast into Eij below.
\begin{align*}
\text{E}22:\;\;
& y q^{-2 (4 b+6 c+1)} Q_{b,c+1}(x,y,z,w)
+q^{-2 (4 b+6 c+1)}(w y z-q^{2 b+4 c+2}) Q_{b+1,c}(x,y,z,w)\\
&+(w-1) (y-1) Q_{b,c}(x,q^{-2}y,z,q^{-2}w)
+w y (z-1) q^{-2 b} Q_{b,c}(x,y,q^{-4}z,w)=0,\\ 
\text{E}23:\;\; 
& -q^{-2 (4 b+6 c+1)} Q_{b,c+1}(x,y,z,w)
-w z q^{-2 (4 b+6 c+1)} Q_{b+1,c}(x,y,z,w) \\
&+w x (y-1) z q^{-2 (b+2 c)} Q_{b,c}(x,q^{-2}y,z,w)
+(w-1) (x-1) Q_{b,c}(q^{-4}x,y,z,q^{-2}w)\\
&+w (x-1) (z-1) q^{-2 b}Q_{b,c}(q^{-4}x,q^2 y,q^{-4}z,w)=0,\\ 
\text{E}24:\;\; 
&(w y z q^{-2 (b+2 c)}-1) Q_{b,c}(x,y,z,w)
+(1-w) Q_{b,c}(x,y,z,q^{-2}w)\\
&-w (z-1) q^{-2 b}Q_{b,c}(x,q^2 y,q^{-4}z,w)
-w (y-1) z q^{-2 (b+2 c)} Q_{b,c}(q^4 x,q^{-2}y,z,w)=0,
\end{align*}
\begin{align*}
\text{E}32:\;\; 
&q^{-6 b-8 c}(q^{2 (b+2 c)}-w y z)(q^{4 (b+c)}-x y^2 z) Q_{b,c}(x,y,z,w)\\
&-y(q^{2 b}-1) q^{-8 (b+c)}(q^{4 (b+c)}-x y^2 z) Q_{b-1,c+1}(x,y,z,w)\\
&-y z q^{-8 (b+c)} Q_{b+1,c}(x,y,z,w)+(y-1)Q_{b,c}(x,q^{-2}y,z,w)\\
&+y (z-1) q^{-2 b} Q_{b,c}(x,y,q^{-4}z,q^2 w)=0,\\ 
\text{E}33:\;\; 
&w z q^{-6 b-8 c}(q^{4 (b+c)}-x y^2 z) Q_{b,c}(x,y,z,w)\\
&+(q^{2 b}-1) q^{-8 (b+c)}(q^{4 (b+c)}-x y^2 z)Q_{b-1,c+1}(x,y,z,w)
+z q^{-8 (b+c)}Q_{b+1,c}(x,y,z,w)\\
&+(x-1) Q_{b,c}(q^{-4}x,y,z,w)
+x (y-1) z q^{-2 (b+2 c)}Q_{b,c}(x,q^{-2}y,z,q^2 w)\\
&+(x-1) (z-1) q^{-2 b}Q_{b,c}(q^{-4}x,q^2 y,q^{-4}z,q^2 w)=0,\\ 
\text{E}34:\;\; 
&(y z q^{-2 (b+2 c)}-1)Q_{b,c}(x,y,z,w)\\
&+z(q^{4 c}-1) q^{-2 (b+2 c+1)}(q^{2 (b+2 c)}-q^2 w y z)Q_{b+1,c-1}(x,y,z,w)\\
&+(q^{2 b}-1) q^{-2 b}(q^{2 (b+2 c-1)}-w y z)(q^{4 (b+c-1)}-x y^2 z)
Q_{b-1,c}(x,y,z,w)\\
&-(z-1) q^{-2 b} Q_{b,c}(x,q^2 y,q^{-4}z,q^2 w)
-(y-1) z q^{-2 (b+2 c)} Q_{b,c}(q^4x,q^{-2}y,z,q^2 w)=0,\\ 
\text{E}35:\;\; 
&-w z(q^{4 c}-1) Q_{b+1,c-1}(x,y,z,w)-Q_{b,c}(x,y,z,q^2 w)\\
&-w(q^{2 b}-1) q^{4 c}(q^{4 (b+c-1)}-x y^2 z)Q_{b-1,c}(x,y,z,w)
+Q_{b,c}(x,y,z,w)=0,
\end{align*}
\begin{align*}
\text{E}42:\;\; 
&x y^2(q^{2 b}-1) Q_{b-1,c+1}(x,y,z,w)
+x y q^{2 b}(w y z-q^{2 (b+2 c)}) Q_{b,c}(x,y,z,w)\\
&-(w-1) q^{6 b+8 c}Q_{b,c}(x,y,z,q^{-2}w)-Q_{b+1,c}(x,y,z,w)=0,\\ 
\text{E}43:\;\; 
&x y^2(1-q^{2 b}) Q_{b-1,c+1}(x,y,z,w)
+w x y q^{2 b}(q^{2 (b+2 c)}-y z) Q_{b,c}(x,y,z,w)\\
&-(w-1) x (y-1) q^{6 b+4c} Q_{b,c}(x,q^{-2}y,q^4 z,q^{-2}w)\\
&-(w-1) (x-1) q^{6 b+8 c} Q_{b,c}(q^{-4}x,q^2y,z,q^{-2}w)
+Q_{b+1,c}(x,y,z,w)=0,\\ 
\text{E}44:\;\; 
&(q^{4 c}-1)(q^{2 (b+2 c-1)}-w y z) Q_{b+1,c-1}(x,y,z,w)
-w y q^{-2 b} Q_{b,c}(q^4 x,y,z,w)\\
&+x y^2(q^{2 b}-1) q^{4 c}(w y z-q^{2 (b+2 c-1)})Q_{b-1,c}(x,y,z,w)
+(w-1) q^{4 c} Q_{b,c}(x,q^2y,z,q^{-2}w)\\
&+(w-1)(y-1)Q_{b,c}(q^4 x,q^{-2}y,q^4 z,q^{-2}w)
+y Q_{b,c}(x,y,z,w)=0,\\
\text{E}45:\;\; 
&w x y^2 z(q^{2 b}-1) q^{4 c} Q_{b-1,c}(x,y,z,w)
-w z(q^{4 c}-1) Q_{b+1,c-1}(x,y,z,w)\\
&-w q^{-2 b}Q_{b,c}(x,q^2 y,z,w)
+(w-1) Q_{b,c}(x,y,q^4 z,q^{-2}w)+Q_{b,c}(x,y,z,w)=0,
\end{align*}
\begin{align*}
\text{E}53:\;\; 
&-x y q^{-2 (b+2 c)}Q_{b,c}(x,y,z,q^2 w)
+(x-1) Q_{b,c}(q^{-4}x,q^2 y,z,w)\\
&+x (y-1) q^{-4 c}Q_{b,c}(x,q^{-2}y,q^4 z,w)+Q_{b,c}(x,y,z,w)=0,\\ 
\text{E}54:\;\; 
&y(q^{2 b}-1) Q_{b-1,c}(x,y,z,w)
+q^{2 b}(q^{4 c}-1)(q^{2 (b+2 c-2)}-w y z)Q_{b,c-1}(x,y,z,w)\\
&-q^{-4 b-4 c+6} Q_{b,c}(x,q^2 y,z,w)
+y q^{-6 b-8 c+6} Q_{b,c}(q^4 x,y,z,q^2 w)\\
&-(y-1) q^{-4b-8 c+6} Q_{b,c}(q^4 x,q^{-2}y,q^4 z,w)=0,\\ 
\text{E}55:\;\; 
&-w z q^{2 b}(q^{4 c}-1)Q_{b,c-1}(x,y,z,w)
+(q^{2 b}-1) Q_{b-1,c}(x,y,z,w)\\
&-q^{-4 b-8 c+6}Q_{b,c}(x,y,q^4 z,w)+q^{-6 b-8 c+6} Q_{b,c}(x,q^2 y,z,q^2 w)=0.
\end{align*}

\end{document}